\documentclass[lettersize,journal,onecolumn]{IEEEtran}
\usepackage{amsmath,amsfonts}
\usepackage{algorithmic}
\usepackage{algorithm}
\usepackage{array}
\usepackage[caption=false,font=normalsize,labelfont=sf,textfont=sf]{subfig}
\usepackage{textcomp}
\usepackage{stfloats}
\usepackage{url}
\usepackage{verbatim}
\usepackage{graphicx}
\usepackage{cite}
\usepackage{hyperref}
\usepackage{xcolor}
\usepackage{siunitx}

\usepackage{comment}

\usepackage{amssymb,amsmath,amsfonts,epsfig,mdframed,rotating,cancel,mathtools}
\usepackage{algorithm}
\usepackage{algorithmic}
\usepackage{textcomp}

\usepackage{xcolor}
\usepackage{tabularx}
\usepackage{multirow}
\usepackage{booktabs}
\usepackage{hyperref}
\usepackage{graphicx}
\usepackage{subcaption}
\usepackage{multirow}
\usepackage{longtable}
\usepackage{bbm}
\usepackage{comment} 

\usepackage[utf8]{inputenc}

\mathcode`l="8000
\begingroup
\makeatletter
\lccode`\~=`\l
\DeclareMathSymbol{\lsb@l}{\mathalpha}{letters}{`l}
\lowercase{\gdef~{\ifnum\the\mathgroup=\m@ne \ell \else \lsb@l \fi}}%
\endgroup

\def\N{{\mathbb N}}

\def\R{{\mathbb R}}

\newfont{\bbb}{msbm10 scaled 500}

\newfont{\bb}{msbm10 scaled 1100}

\newcommand{\av}{{\bf a}}
\newcommand{\bv}{{\bf b}}

\newcommand{\gv}{{\bf g}}

\newcommand{\pv}{{\bf p}}

\newcommand{\uv}{{\bf u}}

\newcommand{\vv}{{\bf v}}
\newcommand{\xv}{{\bf x}}
\newcommand{\yv}{{\bf y}}
\newcommand{\zv}{{\bf z}}

\newcommand{\onev}{{\mathbbm 1}}

\newcommand{\Am}{{\bf A}}
\newcommand{\Bm}{{\bf B}}

\newcommand{\Dm}{{\bf D}}

\newcommand{\Hm}{{\bf H}}
\newcommand{\Id}{{\bf I}}
\newcommand{\Jm}{{\bf J}}

\newcommand{\Lm}{{\bf L}}
\newcommand{\Mm}{{\bf M}}

\newcommand{\Pm}{{\bf P}}

\newcommand{\Wm}{{\bf W}}

\newcommand{\Xm}{{\bf X}}
\newcommand{\Ym}{{\bf Y}}
\newcommand{\Zm}{{\bf Z}}

\newcommand{\Cc}{{\cal C}}
\newcommand{\Dc}{{\cal D}}

\newcommand{\thetav}{\hbox{\boldmath$\theta$}}

\newcommand{\varepsilonv}{\hbox{\boldmath$\varepsilon$}}

\newcommand{\Pim}{\hbox{\boldmath$\Pi$}}

\renewcommand{\l}{(l)}

\newcommand{\diag}{{\hbox{diag}}}

\renewcommand{\arg}{{\hbox{\rm arg}}}

\newcommand{\<}{\left\langle}
\renewcommand{\>}{\right\rangle}

\def\<{\langle}
\def\>{\rangle}

\newcommand{\Diag}{\ensuremath{\operatorname{Diag}}}

\newcommand{\dom}{\ensuremath{\mathrm{dom}\,}}

\usepackage{tikz,tkz-tab}

\usetikzlibrary{fit,automata,positioning,calc,shapes}
\usepackage{ifthen}

\ifdefined\circlenode
\else
\newcommand{\circlenode}[3]{
	\node [draw, circle, inner sep=2pt, #1] (#2) {#3};
}
\fi
\ifdefined\sqnode
\else
\newcommand{\sqnode}[3]{
	\node [draw, rectangle, inner sep=2pt, #1] (#2) {#3};
}
\fi

\usepackage{amsthm} 

\definecolor{darkgreen}{rgb}{0.1,0.5,0}
\definecolor{Midnightblue}{rgb}{.2,.2,.7}
\definecolor{mypink3}{cmyk}{0, 0.7808, 0.4429, 0.1412}

\definecolor{amber}{rgb}{1.0, 0.49, 0.0}

\newcommand \remove[1]{ }

 \newtheorem{definition}{Definition} 
 \newtheorem{proposition}{Proposition}
 
\DeclareMathOperator{\CosSim}{S_C}

\begin{document}

\title{Majorization-Minimization Networks for Inverse Problems:\\ An Application to EEG Imaging}
\author{
  Le~Minh~Triet~Tran, Sarah~Reynaud,
  Ronan~Fablet, Adrien~Merlini, François~Rousseau, and Mai~Quyen~Pham%
  \thanks{
  This work was supported by the European Innovation Council (EIC) through the European Union's Horizon Europe research Programme under Grant 101046748 (Project CEREBRO). This work was granted access to the HPC resources of IDRIS under the allocation 2023-AD010314332 made by GENCI.}%
  \thanks{Le Minh Triet Tran (corresponding author), Sarah Reynaud, and François Rousseau are with \textit{IMT Atlantique}, \textit{LaTIM U1101 INSERM}, Brest, France (emails: \{le-minh-triet.tran, sarah.reynaud, francois.rousseau\}@imt-atlantique.fr).}%
  \thanks{Ronan Fablet, Adrien Merlini, and Mai Quyen Pham are with \textit{IMT Atlantique}, \textit{Lab-STICC UMR CNRS 6285}, Brest, France (emails: \{ronan.fablet, adrien.merlini, mai-quyen.pham\}@imt-atlantique.fr).}%
}

\maketitle

\begin{abstract}
Inverse problems are often ill-posed and require optimization schemes with strong stability and convergence guarantees. While learning-based approaches such as deep unrolling and meta-learning achieve strong empirical performance, they typically lack explicit control over descent and local geometry, particularly the curvature information of the loss, limiting robustness. We propose a learned Majorization-Minimization (MM) framework for inverse problems in a bilevel setting. Instead of learning a full optimizer, we learn a structured curvature majorant that governs each MM step while preserving classical MM descent guarantees. The majorant is parameterized by a lightweight recurrent neural network and explicitly constrained to satisfy valid MM conditions. For cosine-similarity, we derive explicit curvature bounds yielding diagonal majorants, while general objectives can be accommodated via HVP-based curvature estimation. Experiments on EEG source imaging demonstrate improved accuracy, stability, and cross-dataset generalization over deep-unrolled and meta-learning baselines.
\end{abstract}

\begin{IEEEkeywords}
Bilevel Optimization, Learned Majorization-Minimization, Recurrent Neural Networks, EEG Source Imaging.
\end{IEEEkeywords}

\section{Introduction}
\IEEEPARstart{E}{stimating} a phenomenon from indirect or noisy measurements leads to an inverse problem, which is typically ill-posed and requires appropriate priors to ensure stable recovery~\cite{engl1996regularization}. Such problems arise in numerous applications, including telecommunications, geoscience, imaging, and biomedical signal analysis~\cite{Malioutov2005sensorarrays, Fablet2021FourDVarNet, lustig2008compressedMRI, grech2008review}. A particular challenge in realistic environments is the presence of uncontrolled amplitude attenuation, such as fading effects in antenna arrays or skull conductivity attenuation in EEG Source Imaging (ESI)~\cite{Malioutov2005sensorarrays, grech2008review}. In this context, cosine similarity has emerged as a highly effective data-fidelity measure due to its scale-invariant properties. However, the non-convex and non-linear geometry of the cosine loss makes traditional first-order optimizers, such as Gradient Descent, prone to local minima or slow convergence.

To accelerate convergence, second-order methods like Newton's algorithm utilize the inverse Hessian to navigate based on the objective's curvature~\cite{Nocedal2006Numerical}. While mathematically ideal, the $\mathcal{O}(N^3)$ computational cost of Newton-type methods makes them impractical for large-scale high-resolution problems. Quasi-Newton variants like L-BFGS attempt to approximate the Hessian, yet they often struggle to maintain positive definiteness in severely ill-conditioned scenarios. A related yet more general class of approaches is provided by quadratic Majorization-Minimization (MM) methods, which perform updates based on a quadratic surrogate majorant, guaranteeing monotonic descent and rigorous stability~\cite{hunter2000ot}. Nevertheless, the core bottleneck of classical MM remains the design of analytical bounds, which are often overly conservative and limit practical convergence speed.

\IEEEpubidadjcol

In this work, we propose an analytically-constrained framework, termed MM-Net, designed to resolve this bottleneck. Trained via a supervised bilevel optimization scheme, MM-Net abandons hand-crafted static priors in favor of a learned deep prior. Concurrently, we restructure the optimization engine by employing a Recurrent Neural Network (RNN) to learn an adaptive preconditioner. This network approximates second-order curvature information, allowing the algorithm to achieve significantly accelerated convergence. Crucially, this learned component does not operate as a black box. The output of the RNN is projected using the analytical spectral bounds of the cosine similarity loss, which we derive in this paper, ensuring every update remains beneath the analytical ``safety ceiling" of the MM principle. This design seamlessly combines the acceleration power of deep learning with the rigorous descent guarantees of classical optimization.

While MM-Net is inherently domain-agnostic, we deliberately validate it on ESI. Rather than conducting generic multi-domain evaluations, we select ESI because its severely ill-conditioned leadfield matrix $\Lm$ and massive null-space (since sources vastly outnumber sensors) represents a challenging scenario for inverse problems. To demonstrate that MM-Net learns generalizable optimization dynamics rather than dataset-specific priors, we subject it to rigorous stress tests:

\begin{itemize}
    \item \textbf{High-Fidelity In-Domain Resolution:} Validating reliable inversion under both mathematically well-behaved synthetic signals (SEREEGA) and complex biological dynamics (Neural Mass Models)~\cite{sun2022deep}
    \item \textbf{Cross-Paradigm Generalization:} Demonstrating zero-shot transferability between SEREEGA and NMM. This confirms that analytical MM majorants safeguard the solver against severe domain shifts, avoiding the collapse commonly observed in unconstrained networks.
\end{itemize}

The paper is organized as follows. Section~\ref{sec:related} reviews related work. Section~\ref{sec:prelim} formulates the ill-posed inverse problem studied in this paper and recalls the main principles and conditions of quadratic MM. The proposed methodology is presented in Sections~\ref{sec:method} and~\ref{sec:bilevel}. In Section~\ref{sec:method}, we address the construction of valid quadratic MM majorants for cosine-similarity losses with analytical upper bounds. Section~\ref{sec:bilevel} then introduces the bilevel optimization framework based on the proposed MM solver. Section~\ref{sec:appli} describes the application setting and datasets. Section~\ref{sec:resul} reports experimental results on the ESI problem. Finally, Section~\ref{sec:conclu} concludes the paper and outlines directions for future work.

\textbf{Notation}:  
Uppercase bold letters (e.g., $\Lm$) denote matrices, and lowercase bold letters (e.g., $\xv$) denote column vectors. The symbol $\odot$ denotes the Hadamard (entry-wise) product. $\Diag(\vv)$ denotes the diagonal matrix with entries defined from vector $\vv$. 
 $\Dm \succeq 0$ (resp. $\Dm \succ 0$) means that $\Dm \in \R^{n\times n}$ is positive semi-definite, SPD (resp. positive definite, PD), that is, for all $\xv \in \R^n$, $ \xv^\top \Dm \xv \geq 0$  (resp. $ \xv^\top \Dm \xv > 0$). $\onev_{n}$  denotes the vector of ones with length $n$ and $\onev_{m,n} = \onev_m \onev_n^\top$. 

\section{Related works} \label{sec:related}

\paragraph{Direct Mapping vs. Iterative Solvers}
The rise of deep learning has introduced Direct Mapping or End-to-End approaches, where networks are trained to predict the solution directly from measurements in a single forward pass~\cite{Jin2017DCNN, reynaud2024comprehensive}. Despite their fast inference, these black-box models often ignore the forward operator, leading to results that may be inconsistent with the physical measurements and prone to performance degradation under out-of-distribution noise. To enforce physical constraints, the iterative optimization remains a principled standard, as it continuously evaluates data fidelity to ensure the solution adheres to the underlying physics.

\paragraph{Classical Optimization: From First-Order to MM}
Iterative solvers typically rely on first-order algorithms, such as Proximal Gradient Descent (PGD) and Conjugate Gradient (CG), due to their low computational cost and flexibility. However, in complex and ill-conditioned geometries, second-order or Quasi-Newton methods~\cite{Nocedal2006Numerical} are necessary to navigate the curvature for high-precision recovery. Within this landscape, MM offers a principled way to replace difficult objectives with a sequence of tractable surrogates~\cite{hunter2000ot, dempster1977em}. Nevertheless, the integration of scale-invariant metrics like cosine similarity into MM has historically been hindered by a lack of spectral curvature bounds, often resulting in sub-optimal update rules.

\paragraph{PnP Priors and Unrolling}
Building upon first-order frameworks, Plug-and-Play (PnP) methods leverage deep learning by replacing proximal operators in algorithms like PGD with pre-trained denoisers. While frameworks like RED~\cite{romano2017RED} offer theoretical convergence guarantees, they require analytical properties that modern deep networks rarely satisfy, making such guarantees difficult to ensure in practice~\cite{reehorst2018RED2}. To dynamically adapt the optimization loop, unrolling architectures restructure iterative steps into fixed-depth neural networks trained end-to-end. This paradigm has evolved from unrolling classical first-order methods (e.g., MoDL~\cite{Aggarwal2019MoDL}, FISTA-Net~\cite{Xiang2021FISTANet}) to recent MM formulations, such as Unrolled Half-Quadratic (U-HQ) networks~\cite{grarbi2024MMHQ}. However, these unrolled networks lack generalization beyond their finite training horizon and may compromise convergence guarantees by substituting analytical rules with learned parameters, risking divergence in ill-posed scenarios~\cite{monga2021AlgorithmUnrolling}.

\paragraph{Preconditioning, L2O, and Meta-Learning}
Preconditioning is essential for tackling ill-conditioned inverse problems~\cite{Saad2003Iterative, Benzi2002Preconditioning}. While optimizers like Adam~\cite{Kingma2015Adam} approximate diagonal preconditioning using empirical gradient statistics,  they operate independently of the physical curvature of the forward operator. To learn domain-aware preconditioners, the Learning-to-Optimize (L2O) paradigm emerged as a sub-branch of Learning-to-Learn (L2L)~\cite{Li2017L2O, Chen2022L20, Andrychowicz2016LearningToLearn}. Here, RNNs successfully act as data-driven preconditioners. However, lacking analytical constraints, these black-box L2O models risk severe instability. Similarly, Meta-learning frameworks~\cite{Park2019MetaCurvature, Simon2020ModGrad} primarily optimize neural network initialization rather than ensuring physical descent. We resolve this bottleneck by leveraging the memory states of recurrent networks to track optimization trajectories, predicting an adaptive preconditioner strictly regulated by explicit MM bounds to seamlessly bridge L2O with guaranteed physical descent.

\paragraph{Beyond Primal Gradient-Based Domains}
Existing literature also features architectures that operate in dual spaces or altered geometries. Methods such as ADMM-Net~\cite{yang2016deepADMM} and Learned Primal-Dual (LPD)~\cite{adler2018LearnedPrimalDual} utilize operator splitting and primal-dual formulations to handle complex constraints, while Learnable Mirror Descent (LMD)~\cite{Tan2023MirrorDescent} adapts the optimization geometry via Bregman divergences. While effective, these methods optimize in augmented or transformed domains, making primal stability and curvature control harder to characterize theoretically. In contrast, we focus strictly on resolving bottlenecks within the original primal gradient-based domain. By maintaining the original geometry, we isolate and prove the core value of curvature control via analytical majorization, a stability mechanism that space-expanding methods are not primarily designed to address.

\paragraph{Generative Models and the Hallucination Risks}
Recently, generative models such as Diffusion Models and Flow Matching~\cite{Ho2020Diffusion, lipman2023flow} have set new standards for learning complex priors. Methods like Diffusion Posterior Sampling (DPS)~\cite{chung2023DPS} alternate between generative steps and data-consistency projections. Despite their visual quality, a critical drawback is the risk of severe hallucinations. In underdetermined systems like ESI, the forward operator has a massive null-space. Projections only constrain the row-space, leaving the null-space to be freely manipulated by the generative prior~\cite{Kawar2022DDRM}. In contrast, our focus is strictly on deterministic MAP estimation, where MM-Net employs an AutoEncoder (AE) to learn a structural manifold, but critically, it explicitly couples this prior with the analytical MM dynamics. By locking the iterative updates within the MM guardrails, we ensure that every state transition strictly follows the physical gradient, structurally preventing the injection of unverified or hallucinated artifacts into the null-space.

\section{Preliminaries}\label{sec:prelim}
\subsection{Problem Setup}
\label{pbset}
When solving inverse problems, the goal is to estimate an unknown signal $\xv \in \R^s$ from measurements $\yv \in \R^n$ through an observation model:
\begin{equation} \label{eq:invpb}
    \yv = \Lm \xv + \varepsilonv,
\end{equation}
where $\varepsilonv \in \mathbb{R}^{n}$ represents measurement noise or modeling error, and $\Lm \in \mathbb{R}^{n \times s}$ is a linear forward operator. In practice, $\Lm$ is often non-invertible (typically $s > n$) and noise-sensitive, making~\eqref{eq:invpb} an ill-posed problem.

To ensure well-posedness of the analysis, we restrict the measurements to signals with bounded non-zero energy. More precisely, we assume that there exist two constants $\overline{\delta} \ge \underline{\delta} > 0$ such that
$ \yv \in \Cc = \{\yv \in \R^n : \underline{\delta} \le \|\yv\| \le \overline{\delta}\}.$
This assumption is not restrictive in practice, since null measurements do not provide meaningful information for the considered inverse problems and realistic acquisition systems naturally produce signals with non-zero energy.

The unknown signal $\xv$ can be estimated by solving a regularized variational problem of the form:
\begin{equation}\label{eq:optpb}
    \hat{\xv}(\thetav) 
    = \arg\min_{\xv \in \R^s} 
    u(\xv; \yv, \thetav)
    = f \left(\yv, \Lm\xv \right) 
    + \lambda\, (h \circ\Phi)(\xv;\thetav).
\end{equation}

This formulation follows the standard strategy in inverse problems, stabilizing an ill-posed reconstruction using appropriate prior information. To mirror the classical structure, the cost function is decomposed into two components:

\begin{itemize}
    \item \textbf{Data fidelity}\\
    The function $f: \R^n \times \R^n \to \R$ measures the discrepancy between the measurements  $\yv$ and their reconstruction $\Lm\xv$. We assume that  $f$ is twice differentiable with a $\mu_f$-Lipschitz gradient.\footnote{A twice differentiable function $f$ is a $\mu$-Lipschitz gradient if there exists a constant  $\mu< +\infty$ such that $\|\nabla f(\xv)- \nabla f(\yv)\| \leq \mu \|\xv -\yv\|,\; \forall  \xv, \yv \in \R^n.$ }
    \item \textbf{Regularization}\\
    The prior is encoded through the composition $h \circ\Phi$ where \\
    $\Phi: \R^s \to \R^m$ is a (possibly nonlinear) representation operator, twice differentiable and parameterized by $\thetav$, and\\
    $h: \R^m \to \R$ is a differentiable function promoting desirable properties in the recovered sources. We assume that h has a $\mu_h$-Lipschitz continuous gradient.
    \item \textbf{Trade-off parameter}\\
    The scalar $\lambda > 0$ balances data fidelity against the regularization term.
\end{itemize}

This framework captures a broad class of regularized inverse problems, where the choices of $f$, the representation $\Phi$, the prior $h$, and the regularization weight $\lambda$ crucially influence the recovered solution. As widely documented \cite{Vogel2002, arridge2019SolvingInverseProblems, AntilDiKhatri2020}, selecting these components, particularly designing an effective representation $\Phi$ or tuning $\lambda$, remains a central challenge, as these choices are highly problem-dependent and sensitive to both the data and the application at hand.

We propose computing $\hat{\xv}(\thetav)$ via a solver that generates a sequence $\hat{\xv}^0(\thetav), \hat{\xv}^1(\thetav), \ldots$ grounded in the MM strategy.

\subsection{Quadratic MM and Majorant Condition}

MM methods solve a minimization problem by iteratively replacing a difficult-to-optimize function with a simpler upper-bounding surrogate. In our context, we construct a quadratic majorant for the objective in~\eqref{eq:optpb}. We recall the fundamental definition below.

\begin{definition}
Let $\xi: \R^s \to \R$ be a differentiable function, and let $\overline{\xv} \in \R^s$ be a reference point. For any $\xv \in \R^s$, define the quadratic surrogate:
\begin{equation}
    Q(\xv, \overline{\xv})
    = \xi(\overline{\xv})
    + (\xv - \overline{\xv})^\top \nabla \xi(\overline{\xv})
    + \frac{1}{2} (\xv - \overline{\xv})^\top
      \Pm_{\overline{\xv}}
      (\xv - \overline{\xv}),
\end{equation}
where $\Pm_{\overline{\xv}} \in \R^{s \times s}$ is a symmetric positive definite (SPD) matrix. The matrix $\Pm_{\overline{\xv}}$ is said to satisfy the \emph{majorant condition} at $\overline{\xv}$ if
\begin{equation}
    \xi(\xv) \leq Q(\xv, \overline{\xv}), \qquad \forall\, \xv \in \R^s.
\end{equation}
In this case, $Q(\cdot, \overline{\xv})$ is a quadratic majorant of $\xi$ at~$\overline{\xv}$.
\end{definition}

Since $\Pm_{\overline{\xv}}$ is SPD, minimizing the quadratic surrogate yields a closed-form update:
\begin{equation} \label{eq:MM_update_rule}
    \hat{\xv} 
    = \overline{\xv} 
    - \Pm_{\overline{\xv}}^{-1} \nabla \xi(\overline{\xv}).
\end{equation}

A key benefit of MM is that a well-designed majorant can provide a larger and safer descent step than standard gradient descent. However, using a dense SPD matrix requires computing $\Pm_{\overline{\xv}}^{-1}$, which may be computationally expensive.

To balance computational efficiency and convergence guarantees (i.e. ensuring $\Pm_{\overline{\xv}}$ closely upper-bounds the Hessian matrix), we aim to construct a \emph{diagonal} majorant matrix, which allows for fast inversion while satisfying the majorization condition. Let $\pv \in \mathbb{R}^s$ denote the vector of inverse diagonal curvature coefficients such that $\Pm_{\overline{\xv}}^{-1} = \operatorname{Diag}(\pv)$. The MM update reduces to an element-wise scaling of the gradient:
\begin{equation}\label{eq:MM_diag_update}
    \hat{\xv} = \overline{\xv} - \pv \odot \nabla \xi(\overline{\xv}).
\end{equation}

\section{Learning with Analytical MM Guardrails}\label{sec:method}
Gradient-Lipschitz bounds for classical loss functions, such as the squared $\ell_2$ (MSE) loss, the $\ell_1$ loss associated with Laplace noise, or the KL divergence, are well established in the optimization literature \cite{boyd2004convex,nesterov2004introductory,pham2025secondorder}. In contrast, to our knowledge, no standard Lipschitz gradient bounds are available for the cosine similarity, despite its increasing importance in applications where scale and amplitude invariance are desirable, such as ESI. This limitation becomes problematic when cosine similarity is incorporated into iterative optimization schemes, for example in MM, since explicit curvature bounds are required to construct a valid quadratic majorant. 

To concretize the objective formulated in Eq.~\eqref{eq:optpb}, we define both the data-fidelity term $f$ and the regularization term $h$ using the cosine distance. Given two vectors $\av, \bv \in {\R^*}^m$, the cosine similarity is defined as:
\begin{equation}
    \CosSim(\av, \bv) = \frac{\av^\top \bv}{\| \av\| \| \bv \|} \, .
\end{equation}

Consequently, the objective function to be minimized at each step is explicitly given by:
\begin{equation}\label{eq:cosine_loss}
    u(\xv; \yv, \thetav) 
    = \left( 1 - \CosSim(\yv, \Lm\xv) \right) 
    + \lambda \left( 1 - \CosSim(\xv, \Phi(\xv;\thetav)) \right).
\end{equation}

As mentioned above, to ensure numerical stability, we consider the following mild constraint: $\Dc = \{\xv \in \R^s:\; \|\xv\| \ge \upsilon\}$ for a small $\upsilon>0$. In ESI, this is physically justified: the brain perpetually exhibits spontaneous background activity (even in a resting state), which, combined with unavoidable instrumental noise, establishes a non-vanishing floor in the recorded signal. Consequently, practical iterates never converge to zero, and the $\upsilon$-constraint formalizes a condition that naturally holds in all realistic scenarios.

Under this assumption, together with the bounded energy of the observation data, we show that the cosine similarity loss has a Lipschitz continuous gradient on any bounded convex subset of $\Dc$. Beyond merely ensuring the existence of a majorant, this property allows us to derive a state-dependent analytical bound that serves as an operational ``safety ceiling'' for the solver. By projecting the output of RNN into the feasible region defined by this bound, we ensure that every accelerated step remains strictly within the region of monotonic descent. The following Proposition establishes the conditions and the explicit curvature formula for our MM framework.

\begin{proposition}\label{pro:mm}
For every $\xv(\thetav) \in \R^s$, $\yv\in \R^n$, and $\Lm \in \R^{n\times s}$, define the vector $\mathbf{p}(\xv(\thetav)) \in \R^s$ bounded by:
\begin{equation}\label{eq:propo1}
    \underline{\nu} \mathbf{1} \preceq \mathbf{p}(\xv(\thetav)) \preceq  \frac{1}{\mu_{\text{data}} + \lambda \mu_{\text{prior}}} \mathbf{1},
\end{equation}
where $\underline{\nu}>0$ is a small scalar ensuring positive definiteness, and the curvature components $\mu_{\text{data}}$ and $\mu_{\text{prior}}$ are defined as:
\begin{align*}
    \mu_{\text{data}} &= \frac{2 \|\Lm\|_2^2}{\sqrt{3}\|\Lm\xv\|^2},\\
    \mu_{\text{prior}} &= \left( \frac{2 + \sqrt{3} \beta}{\sqrt{3}} \frac{1}{\|\xv\|^2} \right) \\
    &\quad+ \left( \prod_{k=1}^{N} \|\Wm^k\|_2^2 \right) \left( \frac{2 \beta + \sqrt{3}}{\sqrt{3} \beta} \frac{1}{\|\Phi(\xv;\thetav)\|^2} \right)\\
    &\quad+ \left( \sum_{l=1}^{N-1} \left( \prod_{k=1}^{l-1} \|\Wm^k\|_2^2 \right) \big\| \vv_{res}^l \odot \sigma''(\zv^l) \big\|_\infty \right) \Bigg.
\end{align*}
\begin{align*}
    \text{with } &\beta > 0 \text{ is a tightness parameter},\\
    &\Wm^k \text{ denoting the weight matrix of the } k\text{-th layer}, \\
    &\zv^l \text{ as the pre-activation vectors at layer } l,\\
    &\vv_{res}^l \text{ as the residual vectors at layer } l.
\end{align*}
If $\mathbf{p}(\xv(\thetav))$ satisfies~\eqref{eq:propo1}, then it fulfills the rigorous majorization condition for the loss function defined in~\eqref{eq:cosine_loss}.
\end{proposition}

\begin{proof}
    The detailed derivation of the Hessian spectral bounds and the full proof of this Proposition are provided in Appendices~\ref{app:cossim_bounds} and~\ref{app:phi_residual}.
\end{proof}

\section{The Proposed MM-Net Framework}\label{sec:bilevel}
\subsection{Truncated Bilevel Formulation}

In this work, drawing upon the architectural principles in~\cite{Fablet2021FourDVarNet}, we leverage an AE $\Phi$ to capture the underlying structural manifold of the target signals. By embedding this learned representation into the regularization term, we cast the joint optimization of the structural signal priors and the solver's dynamics as a supervised bilevel optimization problem~\cite{crockett2022bilevel,AntilDiKhatri2020}. The formulation is as follows:
\begin{subequations} \label{eq: bilevel}
    \begin{align}
        \hat{\thetav} &= \arg\min_{\thetav} \mathbb{E}_{(\xv_{\text{gt}}, \yv)}[\ell(\xv_{\text{gt}}, \hat{\xv}(\thetav))] \label{bilevel:upper} \\
        \text{s.t. } \hat{\xv}(\thetav) &= \arg\min_{\xv \in \Dc} u\left(\xv; \yv, \thetav \right), \label{bilevel:lower}
    \end{align}
\end{subequations}

where $(\xv_{\text{gt}},\yv)$ be a set of training data which satisfies \eqref{eq:invpb}. The lower-level problem~\eqref{bilevel:lower} estimates the signal $\hat{\xv}(\thetav)$ by minimizing the variational energy $u$ (defined in Eq.~\eqref{eq:cosine_loss}), and the upper-level problem~\eqref{bilevel:upper} updates the network parameters $\thetav$ by minimizing a meta-loss $\ell$. To align training with our directional geometry, we define the meta-loss $\ell$ as:
\begin{equation}\label{eq:meta_loss}
    \ell(\xv_{\text{gt}}, \hat{\xv}) = \Big( 1 - \CosSim(\xv_{\text{gt}}, \hat{\xv}) \Big) + \Big(1 - \CosSim(\hat{\xv}, \Phi(\hat{\xv};\thetav)) \Big).
\end{equation}

The first term ensures directional alignment with the ground truth $\xv_{\text{gt}}$. The second term, acting on the output of the solver $\hat{\xv}$, forces the estimate to reside strictly on the learned structural manifold. This self-consistency constraint ensures that the prior $\Phi$ remains active and effective even when the ground truth is unavailable, bridging the gap between supervised learning and variational stability.

The convergence of the bilevel procedure during training follows from results in nonconvex optimization with inexact inner solvers. As established in \cite{franceschi2020bilevel}, using a truncated gradient-based inner solver with a fixed number of iterations $I$ still ensures convergence to a stationary point of the upper-level objective. While such truncated schemes often lack guarantees beyond their training horizon~\cite{monga2021AlgorithmUnrolling}, our approach ensures that the learned solver remains physically consistent and stable for any number of iterations. By anchoring the updates within the analytical MM bounds derived in Section \ref{sec:method}, we guarantee monotonic descent in the lower-level objective during both training and inference, regardless of the truncation limit $I$.

\begin{algorithm}
    \caption{Truncated Bilevel Training for MM-Net} \label{algo:BL}
    \textbf{Input:} training set $(\xv_{\text{gt}}, \yv, \Lm)$, let $ I, J \in \N^* $, $\gamma \in (0, 2)$, $\eta > 0$.
    \textbf{Initialize:} $\xv^{0} \in \dom u$\footnotemark and $\thetav^0 \in \dom \ell$\\
    \textbf{Output:} $\thetav^{J}$
    \begin{algorithmic}[1]
    \FOR{$j = 0$ \TO $J-1$}
        \FOR{$i = 0$ \TO $I-1$} \STATE \COMMENT{\textit{Inner: Truncated Solver}}
            \STATE $\gv^i = \nabla_{\xv} u(\xv^i; \yv, \thetav^j)$
            \STATE $\pv^i = \operatorname{proj}_{[\underline{\nu}, \, \overline{\nu}^i]} \big(\widetilde{\pv}(\xv^i,\gv^i; \thetav^j) \big)$
            \STATE $\xv^{i+1} = \xv^i - \gamma \pv^i \odot \gv^i$
        \ENDFOR
        \STATE $\hat{\xv} = \xv^I(\thetav^j)$
        \STATE \COMMENT{\textit{Outer: Meta-Loss Update}}
        \STATE $\ell^j = (1 - \CosSim(\xv_{\text{gt}}, \hat{\xv})) + (1 - \CosSim(\hat{\xv}, \Phi(\hat{\xv}; \thetav^j)))$
        \STATE $\thetav^{j+1} = \operatorname{Update}\big(\thetav^j, \nabla_{\thetav} \ell^j, \eta\big)$
    \ENDFOR
    \end{algorithmic}
\end{algorithm}
\footnotetext{$\dom u = \{ \xv: u(\xv) < + \infty\}$}

\subsection{MM-Net Architecture}\label{sec:archi_bilevel}
Motivated by the limitations of fixed, hand-crafted majorants, our objective is to learn an adaptive curvature model that provably preserves the monotonic descent condition. At iteration $i$, let $\gv^i=\nabla_{\xv} u(\xv^i;\yv,\thetav^j)$ denote the gradient of the variational energy computed via automatic differentiation. Instead of relying on a scalar step-size, we predict an unconstrained diagonal inverse curvature vector $\widetilde{\pv}^i \in \R^s$ using a dual-recurrent architecture:
\begin{equation}\label{eq:P_arch}
    \widetilde{\pv}(\xv^i,\gv^i) = \mathcal{F}_1(\gv^i;\theta_1)\odot \mathcal{F}_2\!\big(\xv^i,\mathcal{F}_1(\gv^i;\theta_1);\theta_2\big).
\end{equation}

Crucially, this design is not a purely empirical heuristic, it is structurally inspired by the analytical majorants derived in Proposition~\ref{pro:mm}. Mathematical curvature bounds typically decouple into constant or first-order operator properties and state-dependent scalings. Mirroring this mathematical dynamic, $\mathcal{F}_1$ acts as a gradient-driven feature extractor, while $\mathcal{F}_2$ acts as a state-dependent modulator. We instantiate both modules using ConvLSTM cells~\cite{Shi2015ConvLSTM} to effectively exploit spatio-temporal memory across the optimization trajectory.

To convert this learned prediction into a mathematically valid preconditioner, $\widetilde{\pv}^i$ must be strictly projected into an analytical safety interval:
\begin{equation}\label{eq:P_clip}
    \mathbf{p}^i = \operatorname{proj}_{[\underline{\nu}, \, \overline{\nu}^i]}\!\left(\widetilde{\mathbf{p}}^{\,i}\right),
\end{equation}

where $\operatorname{proj}_{[\underline{\nu}, \overline{\nu}^i]}$ denotes the element-wise Euclidean projection onto the bounding interval. The lower bound $\underline{\nu} > 0$ ensures the strictly positive definiteness of the majorant. The upper bound $\overline{\nu}^i$ acts as the operational guardrail, dynamically populated using the explicit curvature formula established in~\eqref{eq:propo1}. For broader applicability to general differentiable losses beyond cosine similarity, $\overline{\nu}^i$ can be alternatively bounded by $1 / \hat{\lambda}_{\max}$, where $\hat{\lambda}_{\max}$ is the dominant Hessian eigenvalue estimated via standard Power Iterations~\cite{pearlmutter1994FastExactMultiplicationbytheHessian, golub2013MatrixComputations}. Consequently, the projected vector $\pv^i$ induces a valid diagonal quadratic majorant at the current state $\xv^i$, allowing the relaxed descent step to be safely executed via~\eqref{eq:MM_diag_update}. This completes the Truncated Bilevel Training procedure summarized in Algorithm~\ref{algo:BL}.

\begin{figure}
    \centering
    \includegraphics[width=\linewidth]{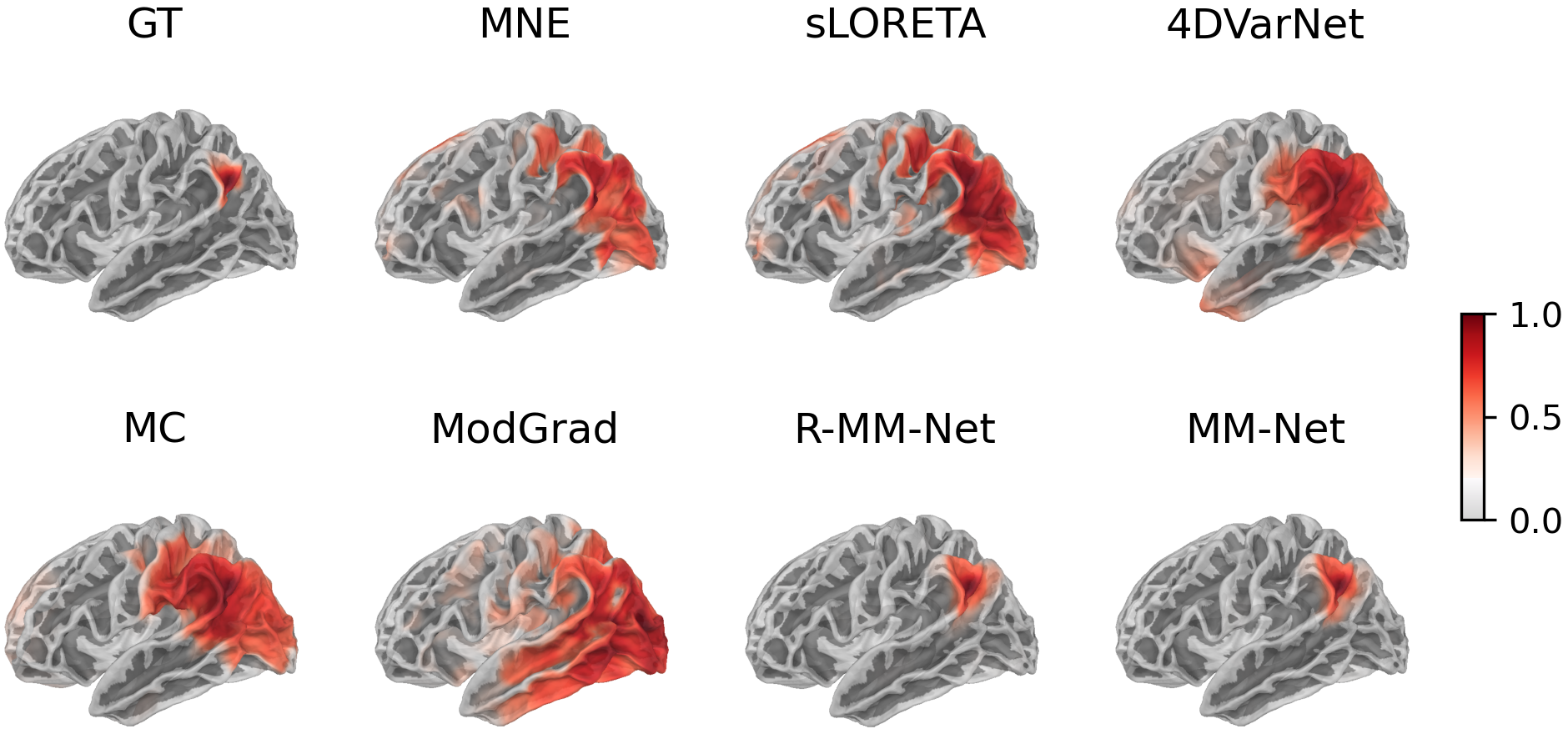}
    \caption{Qualitative reconstructions for a representative sample (SEREEGA). While baselines suffer from severe spatial leakage and spurious artifacts, MM-Net achieves a highly focal reconstruction that closely matches the Ground Truth (GT).}
    \label{fig:cortex}
\end{figure}

\begin{figure}
    \centering
    \includegraphics[width=\linewidth]{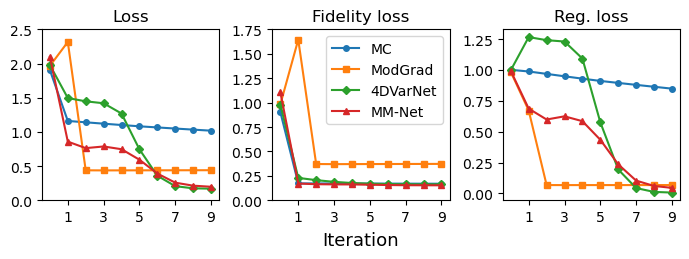}
    \caption{Evolution of the lower-level cost function on the validation set: total loss (loss), data-fidelity loss (fidelity loss), and regularization loss (reg. loss).}
    \label{fig:lw_lv_cost}
\end{figure}

\begin{table*}[t]
    \caption{
        Quantitative comparison of source reconstruction performance. Methods are categorized into Classical, Within-domain, and Cross-domain learning. Best results for each test set are highlighted in bold.
    }\label{tab:dataset_results}
    \centering
    \renewcommand{\arraystretch}{1.05}
    \begin{tabular}{lllcccccc}
    \hline 
    \textbf{Test Set} & \textbf{Type} & \textbf{Method} & \textbf{Params (M)} & \textbf{LE} $\downarrow$ & \textbf{AUC $\uparrow$} & \textbf{nMSE $\downarrow$} & \textbf{PSNR $\uparrow$} & \textbf{TE $\downarrow$} \\
    
    \hline
    \multirow{10}{*}{SEREEGA}
    & \multirow{2}{*}{Classical} & MNE & - & 21.06 & 82.05 & 0.02 & 21.85 & 2.72 \\
    & & sLORETA & - & 14.82 & 86.75 & 0.05 & 18.29 & 1.57 \\
    \cmidrule{2-9}
    
    & \multirow{5}{*}{Within} & Meta-Curvature & 1.80 & 10.22 & 97.86 & 0.01 & 27.90 & \textbf{0.72} \\
    & & ModGrad & 1.34 & 28.68 & 69.12 & 0.06 & 21.54 & 1.60 \\
    & & 4DVarNet & 3.01 & 9.17 & 99.52 & 0.01 & 30.35 & 1.04 \\
    & & R-MM-Net & 2.95 & 4.07 & 99.86 & 0.00 & 37.81 & 0.85 \\
    & & \textbf{MM-Net} & 2.95 & \textbf{3.92} & \textbf{99.90} & \textbf{0.00} & \textbf{38.74} & 0.79 \\
    \cmidrule{2-9}

    & \multirow{3}{*}{\shortstack[l]{Cross \\ \scriptsize (Train: NMM)}} & 4DVarNet & 3.01 & 20.68 & 92.74 & 0.02 & 22.99 & 1.62 \\
    & & R-MM-Net & 2.95 & 17.88 & 92.67 & 0.01 & 25.13 & 1.55 \\
    & & \textbf{MM-Net} & 2.95 & \textbf{16.74} & \textbf{92.87} & \textbf{0.01} & \textbf{25.64} & \textbf{1.47} \\

    \midrule
    \multirow{10}{*}{NMM}
    & \multirow{2}{*}{Classical} & MNE & - & 19.15 & 77.36 & 0.02 & 29.81 & 185.00 \\
    & & sLORETA & - & 14.67 & 79.18 & 0.04 & 26.46 & 165.36 \\
    \cmidrule{2-9}

    & \multirow{5}{*}{Within} & Meta-Curvature & 2.05 & 17.58 & \textbf{84.08} & 0.02 & 30.39 & 170.12 \\
    & & ModGrad & 7.11 & 38.03 & 64.32 & 0.02 & 29.15 & 196.32 \\
    & & 4DVarNet & 3.01 & 15.13 & 80.32 & 0.01 & 32.68 & 166.24 \\
    & & R-MM-Net & 2.95 & 10.33 & 81.14 & 0.01 & 35.72 & \textbf{142.56} \\
    & & \textbf{MM-Net} & 2.95 & \textbf{9.98} & 82.41 & \textbf{0.01} & \textbf{35.81} & 146.24 \\
    \cmidrule{2-9}

    & \multirow{3}{*}{\shortstack[l]{Cross \\ \scriptsize (Train: SEREEGA)}} & 4DVarNet & 3.01 & 14.40 & 74.97 & 0.02 & 30.23 & 156.44 \\
    & & R-MM-Net & 2.95 & 14.20 & 74.47 & 0.01 & 33.25 & 138.78 \\
    & & \textbf{MM-Net} & 2.95 & \textbf{12.48} & \textbf{75.03} & \textbf{0.01} & \textbf{33.72} & \textbf{133.44} \\
    \hline
    \end{tabular}
\end{table*}

\section{Application: ESI problem}\label{sec:appli}
\subsection{Problem: EEG source imaging (ESI)}

As an application for the method proposed in this paper, we focus on the inverse problem appearing in ESI. Electroencephalography (EEG) is a technique in which the electric potentials caused by the electrical activity of the brain are measured at the scalp using an array of electrodes. The relation between the scalp measurements and the brain activity is typically modeled as in Eq.~\eqref{eq:invpb} where $\Ym \in \mathbb{R}^{n \times t}$ contains the potential readings at the $n$ electrodes and $t$ instants, $\Xm \in \mathbb{R}^{s \times t}$ contains the coefficients of the $s$ electric dipoles that model different cluster of neurons at each instant, and $\varepsilon\in \mathbb{R}^{n \times t}$ represents the measurement noise or any other distortion that may impact the EEG signals. The leadfield matrix $\Lm \in \mathbb{R}^{n \times s}$ related the relationship between the $s$ source current dipoles (typically placed on the cortex) and the electric potentials they generate at the $n$ electrodes. It is computed using numerical methods such as the finite element method (FEM) or the boundary element method (BEM) \cite{kybic2005common}. Recovering the source activity from the EEG measurements is an ill-posed inverse problem due to the dimensions and ill-conditioning of $\Lm$.

\subsection{Data simulation}
\paragraph{Head model and forward operator}
All experiments rely on synthetic EEG data in order to ensure the availability of ground-truth source activity and enable quantitative evaluation of the inverse solutions. A common head model and forward operator are used for all datasets. The head model is based on the \textit{fsaverage} MRI template from FreeSurfer, with the cortical surface subsampled into $994$ source locations, following the protocol described in~\cite{sun2022deep}. EEG measurements are simulated using the \textit{standard\_1020} electrode montage with $90$ sensors provided by MNE-Python~\cite{gramfort2013meg}.

The leadfield matrix $\Lm \in \mathbb{R}^{90 \times 994}$ is computed using BEM with three compartments (brain, skull, and scalp) and conductivity values $(0.3,\,0.006,\,0.3)\,\si{\siemens\per\meter\squared}$. This forward model is shared across all simulation settings, while the source temporal dynamics are generated using two different approaches described below.

\paragraph{SEREEGA-based simulations}
The first dataset is generated using the SEREEGA toolbox~\cite{krol2018sereega}, following simulation protocols commonly used in learning-based ESI studies~\cite{sun2022deep}. Source activity consists of single extended cortical sources with Gaussian temporal waveforms. For each sample, a spatially contiguous source patch is formed by aggregating neighboring dipoles around a randomly selected seed, with amplitudes decreasing as a function of geodesic distance. Temporal waveforms follow the Gaussian parametrization implemented in SEREEGA, with parameters drawn from uniform distributions as defined in~\cite{reynaud2024comprehensive}. Each simulated signal contains $64$ time samples. A total of $10{,}000$ samples are generated and split into training and validation sets using an $80\%/20\%$ ratio.

\paragraph{Neural mass model simulations}
To evaluate the proposed method under more realistic and complex source dynamics, a second dataset is generated using a neural mass model (NMM) framework, following the simulation setup described in~\cite{reynaud2025thesis}. In this setting, cortical activity is generated by coupled neural population models, producing physiologically plausible oscillatory signals through nonlinear interactions between excitatory and inhibitory subpopulations.

Spatial source configurations are generated in the same manner as in the SEREEGA dataset, using localized extended patches of neighboring dipoles, ensuring consistent spatial characteristics across datasets. In contrast to the Gaussian waveforms used in SEREEGA, temporal dynamics are governed by the underlying NMM equations, resulting in more complex, nonstationary time courses. Each simulated signal contains $500$ time samples. As with the SEREEGA dataset, $10{,}000$ samples are generated and split into training and validation sets using an $80\%/20\%$ ratio.

Together, these two datasets provide complementary testbeds for evaluating MM-Net: SEREEGA offers a standardized benchmark with controlled temporal structure, while the NMM simulations probe performance under biologically plausible and highly complex dynamics.

\subsection{Algorithm configuration}

For reproducible evaluations, spatio-temporal data matrices of shape $m \times t$ are flattened via vectorization ($\zv = \operatorname{vec}(\Zm) \in \R^{mt}$) to strictly ensure compatibility with our formulation. The meta-loss~\eqref{eq:meta_loss} is evaluated on these vectors and optimized using Adam~\cite{Kingma2015Adam} for $J=500$ epochs with learning rate $\eta = \num{5e-5}$. For the lower-level solver, we set $I=10$ inner iterations to balance computational cost and training resolution. The step-size relaxation is fixed at $\gamma = 1.999$ to preserve parameter-free monotonic descent. To test robustness against poor initializations, the initial state is drawn from a scaled Gaussian, $\Xm^0 \sim 10^{-3}\mathcal{N}(0,1)$, and subsequently vectorized. Finally, the lower-level energy $u$~\eqref{eq:cosine_loss} equally weights the cosine data fidelity and structural prior ($\lambda = 1.0$).

\section{Evaluation}\label{sec:resul}
\subsection{Evaluation metrics}
Following~\cite{reynaud2024comprehensive}, we use five metrics to evaluate the performance of the algorithm on the ESI task: the localization error (LE), the area under the ROC curve (AUC), the normalized mean squared error (nMSE), the peak signal to noise ratio (PSNR) and the time error (TE). The LE indicates how well the method can localize the active source. The AUC is used to quantify the ability to estimate the source extension. The nMSE and PSNR are used to evaluate amplitude reconstruction, and the TE measures the ability of the method to reconstruct the waveform precisely.

\subsection{Results}
To rigorously isolate our core contribution, we evaluate MM-Net against baselines strictly within deterministic, primal-domain optimization, purposefully excluding generative or dual-space methods. Our selection spans three paradigms: (1) \textit{Classical Anchors}: MNE~\cite{hamalainen1994mne} and sLORETA~\cite{pascual2002sloreta} provide established mathematical reference points from classical optimization. (2) \textit{Learned Preconditioning}: Meta-Curvature (MC)~\cite{Park2019MetaCurvature} and ModGrad~\cite{Simon2020ModGrad} represent unconstrained data-driven preconditioners, acting as direct predecessors to our neural optimizer.\footnote{For MC and ModGrad, we reuse the architectures of their learned preconditioners and integrate them into a standard gradient-descent scheme, with careful hyperparameter tuning to ensure competitive performance.} (3) \textit{Unrolled Variational Solvers}: Rather than evaluating generic PnP networks, we select 4DVarNet~\cite{Fablet2021FourDVarNet}. By unrolling a L2O solver coupled with an autoencoder prior, it provides the most architecturally aligned comparison. Finally, \textit{Relaxed MM-Net (R-MM-Net)} serves as an internal stress test by removing the projection mechanism associated with the MM constraints in~\eqref{eq:P_clip}. In the proposed \textit{MM-Net}, the learned diagonal majorization matrix is projected onto analytically derived lower and upper bounds to guarantee the validity of the MM surrogate and preserve the associated stability properties. In contrast, \textit{R-MM-Net} removes this projection step entirely, which is equivalent to projecting onto $(-\infty,+\infty)$ and therefore imposing no constraint on the learned matrix. By comparing both variants, we explicitly evaluate the importance of controlling the learned curvature matrix so that it satisfies the mathematical properties required by the MM framework, and empirically demonstrate that these analytical guardrails are essential for stable domain generalization.

Table~\ref{tab:dataset_results} presents a comprehensive quantitative evaluation across the SEREEGA and NMM test sets. Rather than treating training configurations in isolation, we systematically group methods into classical, within-domain, and cross-domain approaches. This structure allows us to assess not only the absolute reconstruction accuracy under matched conditions but also the models' robustness to significant shifts in temporal dynamics and underlying data generation processes.

\paragraph{Within-domain performance}
The proposed method consistently achieves the best overall performance across all metrics. The learning-based methods outperform the classical approaches (MNE and sLORETA), validating the data-driven paradigm. On SEREEGA, our approach substantially improves LE and nMSE, while yielding higher PSNR and AUC than all baselines. On the more challenging NMM dataset, which features complex and nonstationary temporal dynamics, our method again delivers the lowest LE and nMSE, as well as the highest PSNR. Although MC attains a slightly higher AUC in this setting, it does so at the cost of significantly larger TE, whereas our method achieves a more favorable accuracy-stability trade-off. Furthermore, MM-Net consistently outperforms R-MM-Net ablation on both datasets, confirming that strictly enforcing the rigorous majorization conditions yields empirical accuracy gains.

\paragraph{Cross-domain generalization}
Cross-domain evaluation reveals a more pronounced separation between methods. In both transfer configurations, MM-Net consistently outperforms both 4DVarNet and R-MM-Net across all metrics. Notably, the widening performance gap between MM-Net and the R-MM-Net, proves that theoretical grounding enhances out-of-distribution robustness. In particular, we observe gains in LE, nMSE, and PSNR, indicating that MM-Net generalizes more robustly across datasets with different temporal resolutions and underlying generative processes. These results confirm that MM-Net captures transferable optimization geometry, rather than overfitting to dataset-specific statistics.

MC and ModGrad are excluded from the cross-domain experiments because their learned update rules require a fixed temporal dimensionality, preventing transfer across domains with varying resolutions. This limitation highlights a key advantage of our MM-based formulation, which is inherently agnostic to sequence length.

\paragraph{Qualitative visual evaluation}
To complement the quantitative metrics, Figure~\ref{fig:cortex} visualizes the reconstructed source maps for a representative test sample. While classical and deep-learning baselines suffer from severe spatial leakage and spurious activations, our MM-Net produces highly focal reconstructions. Furthermore, while R-MM-Net and MM-Net yield comparable results, explicitly enforcing the majorant conditions ensures strict theoretical stability without compromising this high spatial fidelity.

\paragraph{Convergence dynamics}
As illustrated in Figure~\ref{fig:lw_lv_cost}, while the MM principle theoretically guarantees the monotonic decrease of the total objective, our proposed MM-Net also exhibits highly stable convergence across both fidelity and regularization components. Unlike ModGrad and 4DVarNet, which suffer from significant initial spikes or slow adaptation, MM-Net maintains a consistent descent trajectory. This empirical stability confirms that the analytical MM guardrails effectively prevent the learned optimizer from exploring divergent or unstable update regions.

\paragraph{Discussion}
Across all evaluation scenarios, MM-Net consistently excels in accuracy, robustness, and stability. By significantly outperforming classical baselines while maintaining mathematical guarantees, our framework successfully bridges the gap between classical optimization and deep learning. The strong cross-dataset performance suggests that learning a structured curvature majorant, rather than a full update rule, yields better generalization in ill-posed inverse problems. Overall, these results confirm that the proposed approach combines the empirical effectiveness of learned optimization with the theoretical robustness of MM, establishing a favorable trade-off between performance and reliability.

\section{Conclusion}\label{sec:conclu}
This work introduces a learned MM framework for ill-posed inverse problems via bilevel optimization. Instead of fully unrolled rules, our approach learns a structured curvature majorant that governs each MM step, while rigorously preserving the descent guarantees of classical MM algorithms. The majorant is parameterized by a lightweight recurrent architecture that adapts to the current iterate and its gradient, thereby combining the transparency and stability of surrogate-based optimization with the flexibility of data-driven learning.

Theoretically, we derived explicit curvature bounds for cosine-similarity losses, leading to simple diagonal quadratic majorants and valid MM updates despite nonlinear learned representations. This benefits scale-invariant applications where recovering the exact signal phase and structural morphology is fundamentally more important than absolute amplitudes. Together, these theoretical foundations yield a provably convergent lower-level solver and guarantee convergence to stationary points within the bilevel framework.

Experimental results on ESI demonstrate that MM-Net consistently outperforms deep-unrolled and meta-learning baselines in accuracy, quality, and robustness. In particular, MM-Net exhibits strong cross-domain generalization across diverse temporal dynamics, regardless of the spatial forward operator. This highlights the benefit of learning transferable curvature models rather than domain-specific update rules. Beyond ESI, this framework extends naturally to other inverse problems. By coupling the learned MM solver with task-specific priors, it provides a principled, transferable approach for broad ill-posed challenges, particularly in scale-invariant applications.

\section*{Acknowledgment}
We utilized Google Gemini and ChatGPT for copy-editing and grammatical polishing throughout to enhance readability. The core scientific content, including all derivations and experimental analyses, remains our original work.

\appendices

\section{Proof of Boundedness of the Objective Function}\label{app:cossim_bounds}
This appendix provides the bounds for the lower objective function~\eqref{eq:cosine_loss}. We first establish a universal spectral bound, then apply this result to derive local majorant matrices for both the data term $\CosSim(\Lm\xv, \yv)$ and the prior term $\CosSim(\xv, \Phi(\xv;\thetav))$.

\subsection*{General Bound of the Projection Metric}
Consider the standard cosine similarity between two vectors $\av, \bv \in \R^{*s}$ defined on a bounded convex subset. The convexity ensures that the line segment $[\av, \bv]$ remains in a region with a bounded Hessian, satisfying the sufficient condition for Lipschitz continuity of the gradient. Let $\uv = \frac{\av}{\|\av\|}$ and $\vv = \frac{\bv}{\|\bv\|}$ be the normalized vectors such that $\CosSim(\av, \bv) = \uv^\top \vv$.

The Jacobians of the $L_2$ normalizations are symmetric, orthogonal projection matrices:
\begin{equation*}
    \Pim_\av = \nabla_\av \uv = \frac{1}{\|\av\|}(\Id - \uv\uv^\top) \; \text{and} \; \Pim_\bv = \nabla_\bv \vv = \frac{1}{\|\bv\|}(\Id - \vv\vv^\top).
\end{equation*}

We compute the first and second derivatives of $\CosSim(\av, \bv)$ with respect to $\av$:
\begin{align*}
    \nabla_\av \CosSim(\av, \bv) &= \Pim_\av \vv = \frac{1}{\|\av\|} (\vv - (\uv^\top \vv)\uv)\\
    \nabla_\av^2 \CosSim(\av, \bv) &= \nabla_\av \left( \frac{1}{\|\av\|} \right) \big(\vv - (\uv^\top \vv)\uv\big)^\top\\
    &\quad +\frac{1}{\|\av\|} \nabla_\av \big(\vv - (\uv^\top \vv)\uv\big) \\
    &= -\frac{1}{\|\av\|^2} \uv \big(\vv^\top - (\uv^\top \vv)\uv^\top\big)\\
    &\quad -\frac{1}{\|\av\|} \Big( \uv (\nabla_\av(\uv^\top \vv))^\top + (\uv^\top \vv) \nabla_\av \uv \Big)\\
    &= \frac{1}{\|\av\|^2} \Mm(\uv, \vv),
\end{align*}
where $\Mm(\uv, \vv) = -(\vv^\top \uv)\Id - \vv\uv^\top - \uv\vv^\top + 3(\vv^\top \uv)\uv\uv^\top$ is the symmetric core matrix.

To find the spectral norm $\|\Mm\|_2$, we determine its eigenvalues. For any vector $\zv$ orthogonal to both $\uv$ and $\vv$, we have $\Mm \zv = -(\vv^\top \uv)\zv$. Thus, $-(\vv^\top \uv)$ is an eigenvalue. 

The remaining non-trivial eigenvectors lie in the 2D subspace spanned by $\{\uv, \vv\}$. We can construct an orthonormal basis $\{\uv, \uv^\perp\}$ for this subspace, where $\vv = c \uv + s \uv^\perp$, with $c = \vv^\top \uv \in [-1, 1]$ and $s = \sqrt{1 - c^2}$. In this basis, the action of $\Mm$ can be represented by the $2 \times 2$ matrix:
\begin{equation*}
    \begin{pmatrix}
        0 & -s \\
        -s & -c
    \end{pmatrix}
\end{equation*}
The characteristic equation of this matrix is $\lambda^2 + c\lambda - s^2 = 0$. Solving for $\lambda$, the two eigenvalues in this subspace are:
\begin{equation*}
    \lambda_{1,2} = \frac{-c \pm \sqrt{c^2 + 4s^2}}{2} = \frac{-c \pm \sqrt{4 - 3c^2}}{2}.
\end{equation*}

The spectral norm $\|\Mm\|_2$ is the maximum absolute eigenvalue. We maximize the function $g(c) = \frac{|c| + \sqrt{4 - 3c^2}}{2}$ over the interval $c \in [-1, 1]$. By setting the derivative to zero, the maximum occurs at $c = \pm \frac{1}{\sqrt{3}}$, yielding:
\begin{equation}\label{eq:core_spectral_bound}
    \|\Mm\|_2 = g\left(\pm \frac{1}{\sqrt{3}}\right) = \frac{2}{\sqrt{3}} \Rightarrow \nabla_\av^2 \CosSim(\av, \bv) \preceq \frac{2}{\sqrt{3}\|\av\|^2} \Id.
\end{equation}

This bound forms the basis for constructing majorants.
\qed

\subsection{Majorant Bound for the Data Term}
Let $\av = \Lm\xv$ and $\bv = \yv$. The first and second derivatives of data term $\CosSim(\Lm\xv, \yv)$ with respect to $\xv$ are:
\begin{align*}
    \nabla_\xv \CosSim(\Lm\xv, \yv) &= \Lm^\top \Pim_{\Lm\xv} \vv,\\
    \nabla_\xv^2 \CosSim(\Lm\xv, \yv) &= \Lm^\top \big(\nabla_\av^2 \CosSim(\av, \yv)\big) \Lm.
\end{align*}

Substituting the bound from Equation~\eqref{eq:core_spectral_bound}, we construct the majorant bound for the Data Term $\CosSim(\Lm\xv, \yv)$:
\begin{equation}\label{eq:b_data}
    \Bm_{data}(\xv) = \frac{2}{\sqrt{3}\|\Lm\xv\|^2} \Lm^\top \Lm.
\end{equation}

For efficient element-wise updates, we relax $\Bm_{data}(\xv)$ into a uniform diagonal majorant bound, guaranteeing monotone descent while completely avoiding matrix inversions. Using the spectral bound $\Lm^\top \Lm \preceq \|\Lm\|_2^2 \Id$, we obtain:
\begin{equation}\label{eq:b_data_diag}
    \widetilde{\Bm}_{data}(\xv) = \frac{2 \|\Lm\|_2^2}{\sqrt{3}\|\Lm\xv\|^2} \Id.
\end{equation}

\textbf{Boundedness Guarantee:} To prevent singularities given $\Lm$'s non-trivial null space, we use Tikhonov regularization $\epsilon \Id$ ($\epsilon > 0$), such that $\|\Lm\xv\|_{\epsilon}^2 = \xv^\top (\Lm^\top \Lm + \epsilon \Id) \xv \ge \epsilon \|\xv\|^2$. For $\xv \neq \mathbf{0}$, this bounds the denominator away from zero, ensuring $\Bm_{data}(\xv)$ remains globally bounded.

\subsection{Majorant Bound for the Prior Term}
Let $\av = \xv$ and $\bv = \Phi(\xv;\thetav)$. Let $\Jm_\Phi = \nabla_\xv \Phi(\xv;\thetav)$ be the Jacobian of $\Phi(\xv;\thetav)$. The first and second derivatives of the prior term $\CosSim(\xv, \Phi(\xv;\thetav))$ with respect to $\xv$ are:
\begin{align*}
    \nabla_\xv \CosSim(\xv, \Phi(\xv;\thetav)) &= \Pim_\xv \vv + \Jm_\Phi^\top \Pim_\Phi \uv,\\
    \nabla_\xv^2 \CosSim(\xv, \Phi(\xv;\thetav)) &= (\nabla_\xv \Pim_\xv) \vv + \Pim_\xv (\nabla_\xv \vv) + \Jm_\Phi^\top \nabla_\xv (\Pim_\Phi \uv)\\
    &\; +\sum_{i=1}^s (\Pim_\Phi \uv)_i \big[\nabla_\xv^2 \Phi(\xv; \thetav)\big]_i\\
    &= \Hm_\xv + \Pim_\xv \Pim_\Phi \Jm_\Phi + \Jm_\Phi^\top \Hm_\Phi \Jm_\Phi\\
    &\; + \Jm_\Phi^\top \Pim_\Phi \Pim_\xv +\sum_{i=1}^s (\Pim_\Phi \uv)_i \big[\nabla_\xv^2 \Phi(\xv; \thetav) \big]_i.
\end{align*}

\paragraph{Gauss-Newton curvature}
Let the Gauss-Newton curvature $\Hm_{GN}(\xv)$ contain only first-order network derivatives. The Jacobian of the network is denoted by $\Jm_\Phi = \nabla_\xv \zv^N \in \R^{s \times s}$, which can be evaluated via the forward recursion defined in Appendix~\ref{app:phi_residual}. The Gauss-Newton term expands to:
\begin{equation*}
    \Hm_{GN}(\xv) = \Hm_\xv + \Pim_\xv \Pim_\Phi \Jm_\Phi + \Jm_\Phi^\top \Pim_\Phi \Pim_\xv + \Jm_\Phi^\top \Hm_\Phi \Jm_\Phi,
\end{equation*}
where $\Hm_\xv$ and $\Hm_\Phi$ are the metric Hessians computed in Equation~\eqref{eq:core_spectral_bound}. Applying the bound to the blocks yields:
\begin{equation*}
    \Hm_\xv\preceq \frac{2}{\sqrt{3}\|\xv\|^2}\Id \quad \text{and} \quad \Hm_\Phi \preceq \frac{2}{\sqrt{3}\|\Phi(\xv;\thetav)\|^2}\Id.
\end{equation*}

To bound the indefinite cross-term, we apply the weighted Young's inequality $\Am^\top \Bm + \Bm^\top \Am \preceq \beta \Am^\top \Am + \frac{1}{\beta} \Bm^\top \Bm$ with a scalar $\beta > 0$. Setting $\Am = \Pim_\xv$ and $\Bm = \Pim_\Phi \Jm_\Phi$, and applying the expanded idempotence property of the projection matrices, we construct the Gauss-Newton majorant bound:
\begin{equation}\label{eq:b_gn}
    \begin{aligned}
        \Bm_{GN}(\xv) &= \left( \frac{2}{\sqrt{3}\|\xv\|^2} \Id + \frac{\beta}{\|\xv\|} \Pim_\xv \right)\\
        &+ \Jm_\Phi^\top \left( \frac{2}{\sqrt{3}\|\Phi(\xv;\thetav)\|^2} \Id + \frac{1}{\beta \|\Phi(\xv;\thetav)\|} \Pim_\Phi \right) \Jm_\Phi.
    \end{aligned}
\end{equation}

For efficient element-wise updates, we relax $\Bm_{GN}(\xv)$ into a uniform diagonal majorant bound. Applying the projection bounds $\Pim_\xv \preceq \frac{1}{\|\xv\|}\Id$, $\Pim_\Phi \preceq \frac{1}{\|\Phi(\xv;\thetav)\|}\Id$, and the Jacobian spectral bound $\Jm_\Phi^\top \Jm_\Phi \preceq \big( \prod_{k=1}^{N} \|\Wm^k\|_2^2 \big) \Id$, we obtain:
\begin{equation}\label{eq:b_gn_diag}
    \begin{aligned}
        \widetilde{\Bm}_{GN}(\xv) &= \Bigg[ \left( \frac{2 + \sqrt{3} \beta}{\sqrt{3}} \frac{1}{\|\xv\|^2} \right) \\
        &\quad+ \left( \prod_{k=1}^{N} \|\Wm^k\|_2^2 \right) \left( \frac{2 \beta + \sqrt{3}}{\sqrt{3} \beta} \frac{1}{\|\Phi(\xv;\thetav)\|^2} \right) \Bigg] \Id.
    \end{aligned}
\end{equation}

Analytically, the tightest diagonal bound is achieved by selecting the optimal $\beta$ that minimizes the trace of $\widetilde{\Bm}_{GN}(\xv)$.

\paragraph{Residual curvature}
Let the residual network curvature $\Hm_{res}(\xv)$ contain second-order network derivatives in $\nabla_\xv^2 \CosSim(\xv, \Phi(\xv;\thetav))$. The residual curvature term expands to:
\begin{equation*}
    \Hm_{res}(\xv) = \sum_{i=1}^s \big(\Pim_\Phi \uv\big)_i \big[\nabla_\xv^2 \Phi(\xv;\thetav)\big]_i.
\end{equation*}

Evaluating this residual Hessian requires prohibitive second-order backpropagation through the deep network. To address this bottleneck, we analyze its layer-wise structure and derive its structural and spectral majorant bounds in Appendix~\ref{app:phi_residual}.

To maximize computational efficiency and enable fully element-wise updates, we construct a uniform diagonal majorant bound for the entire objective~\eqref{eq:cosine_loss}. By aggregating the individual relaxed bounds $\widetilde{\Bm}_{data}(\xv)$~\eqref{eq:b_data_diag}, $\widetilde{\Bm}_{GN}(\xv)$~\eqref{eq:b_gn_diag}, and $\widetilde{\Bm}_{res}(\xv)$~\eqref{eq:b_res_diag}, the total majorant is given by:
\begin{equation}
    \widetilde{\Bm}_{total}(\xv) = \widetilde{\Bm}_{data}(\xv) + \widetilde{\Bm}_{GN}(\xv) + \widetilde{\Bm}_{res}(\xv).
\end{equation}
\qed
 
\section{Residual Majorant for the Autoencoder}\label{app:phi_residual}

We model the $N$-layer autoencoder $\Phi: \R^s \to \R^s$ as a sequence of transformations. Let $\zv^1 = \xv$ and $\Phi(\xv;\thetav) = \zv^N$. For layers $l = 1, \dots, N-1$, the pre- and post-activations are:
\begin{equation*}
    \av^l = \sigma(\zv^l) \; \text{and} \; \zv^{l+1} = \Wm^l \av^l + \bv^l,
\end{equation*}
where $\Wm^l$ and $\bv^l$ are the layer weights and biases. We specify the activation $\sigma$ as the Soft-ReLU function, applied element-wise. Its explicit form and derivatives are given by:
\begin{equation*}
    \begin{aligned}
    \sigma(z) &= \frac{1}{2}\left(z + \sqrt{z^2 + \varepsilon^2}\right), \quad \varepsilon>0, \\
    \sigma'(z) &= \frac{1}{2}\left(1 + \frac{z}{\sqrt{z^2+\varepsilon^2}}\right) \in (0,1), \\
    \sigma''(z) &= \frac{1}{2}\frac{\varepsilon^2}{(z^2+\varepsilon^2)^{3/2}} \le \frac{1}{2\varepsilon}.
    \end{aligned}
\end{equation*}

We define the backward effective weights as:
\begin{align*}
    \vv_{res}^l &= \nabla_{\av^l} \Big(\big( \Pim_\Phi \uv \big)^\top \Phi(\xv;\thetav) \Big)\\
    &=
    \begin{cases}
        (\Wm^l)^\top \big(\Pim_\Phi \uv\big), & \text{if } l =  N - 1,\\
        (\Wm^l)^\top \left( \sigma'(\zv^{l+1}) \odot \vv_{res}^{l+1} \right), & \text{otherwise}.
    \end{cases}
\end{align*}

Applying the multivariable chain rule over the network layers, the exact residual curvature factorizes into:
\begin{align*}
    \Hm_{res}(\xv) &= \sum_{i=1}^s \big(\Pim_\Phi \uv\big)_i \big[\nabla_\xv^2 \Phi(\xv;\thetav)\big]_i\\
    &= \sum_{l=1}^{N-1} \big(\nabla_\xv \zv^l\big)^\top \diag\Big(\vv_{res}^l \odot \sigma''\big(\zv^l\big)\Big) \nabla_\xv \zv^l.
\end{align*}

Since $\Am \preceq |\Am|$ strictly holds for diagonal matrices, we can bound the exact Hessian $\Hm_{res}(\xv)$ by taking the element-wise absolute value. Noting that $\sigma''(\zv) > 0$, we define the structural majorant bound $\Bm_{res}(\xv) \succeq \Hm_{res}(\xv)$ explicitly as:
\begin{equation}\label{eq:b_res}
    \Bm_{res}(\xv) = \sum_{l=1}^{N-1} \big(\nabla_\xv \zv^l\big)^\top \diag\Big(\big|\vv_{res}^l\big| \odot \sigma''\big(\zv^l\big)\Big) \nabla_\xv \zv^l,
\end{equation}
where the forward Jacobian is defined recursively:
\begin{equation*}
    \nabla_{\xv} \zv^l =
    \begin{cases}
        \Id, & \text{if } l = 1, \\
        \Wm^{l-1} \diag\big(\sigma'(\zv^{l-1})\big) \nabla_{\xv} \zv^{l-1}, & \text{for } l > 1.
    \end{cases}
\end{equation*}

For efficient element-wise updates, we relax $\Bm_{res}(\xv)$ into a uniform diagonal majorant bound. Applying the triangle inequality and the sub-multiplicativity of the spectral norm $\|\cdot\|_2$, we obtain the relaxed scalar spectral bound:
\begin{equation}\label{eq:b_res_diag}
    \widetilde{\Bm}_{res}(\xv) = \left( \sum_{l=1}^{N-1} \left( \prod_{k=1}^{l-1} \|\Wm^k\|_2^2 \right) \big\| \vv_{res}^l \odot \sigma''(\zv^l) \big\|_\infty \right) \Id.
\end{equation}

In practice, the weight spectral norms $\|\Wm^k\|_2^2$ are constants that can be precomputed during training. Consequently, evaluating the residual bound at each iteration merely requires computing the infinity norm of the intermediate vectors, thus reducing the computational complexity.
\qed

\bibliographystyle{IEEEtran}
\bibliography{biblio}
\end{document}